\documentclass[UKenglish]{lipics-derivative}

\usepackage[section]{placeins}
\usepackage{authblk}
\usepackage{amsmath}
\usepackage{amssymb}
\usepackage[numbers]{natbib}
\usepackage{tikz}
\usepackage{graphics,graphicx}
\usepackage[usenames,dvipsnames]{pstricks}
\usepackage{hyperref}
\usepackage{color}
\usepackage{tcolorbox} 
\usepackage{doi}

\newtheorem{proposition}{Proposition}
\newtheorem{claim}[theorem]{Claim}

\newcommand{\mL}{\mathcal{L}}
\newcommand{\mV}{\mathcal{V}}
\newcommand{\LCLS}{\mathsf{LCL\textsuperscript{*}}}
\newcommand{\LOCAL}{\mathsf{LOCAL}}
\newcommand{\LCL}{\mathsf{LCL}}
\newcommand{\LD}{\mathsf{LD}}
\newcommand{\EE}{\mathbb{E}}
\newcommand{\PP}{\mathbb{P}}

\title{How long it takes for an ordinary node with an ordinary ID to output?}

\author[1]{Laurent Feuilloley}

\affil[1]{Institut de Recherche en Informatique Fondamentale (IRIF),\\ CNRS and University Paris Diderot,\\ France}

\begin{document}

\maketitle

\begin{abstract}

In the context of distributed synchronous computing, processors perform in rounds, and the time complexity of a distributed algorithm is classically defined as the number of rounds before all computing nodes have output. 
Hence, this complexity measure captures the running time of the slowest node(s). In this paper, we are interested in the running time of the ordinary nodes, to be compared with the running time of the slowest nodes. 
The \emph{node-averaged} time complexity of a distributed algorithm on a given instance is defined as the average, taken over every node of the instance, of the number of rounds before that node outputs. 
We compare the node-averaged time complexity with the classic one in the standard  $\LOCAL$ model for distributed network computing. We show that there can be an exponential gap between the former and the later, as witnessed by, e.g., leader election.  Our first main result is a positive one, stating that, in fact, the two time complexities behave the same for a large class of problems on very sparse graphs. In particular, we show that, for $\LCL$ problems on cycles, the node-averaged time complexity is of the same order of magnitude as the ``slowest node'' time complexity. 
In addition, in the $\LOCAL$ model, the time complexity is computed as a worst case over all possible identity assignments to the nodes of the network. In this paper, we also investigate the \emph{ID-averaged} time complexity, when the number of rounds is averaged over all possible identity assignments of $O(\log n)$-size identifiers, where $n$ is the size of the network. Our second main result is that the ID-averaged time complexity is essentially the same as the expected time complexity of \emph{randomized} algorithms (where the expectation is taken over all possible random bits used by the nodes, and the number of rounds is measured for the worst-case identity assignment). 
Finally, we study the node-averaged ID-averaged  time complexity. We show that 3-colouring the $n$-node ring requires $\Theta(\log^*\!n)$ rounds if the number of rounds is averaged over the nodes, or if the number of rounds is averaged over the identity assignments. In contrast, we show that 3-colouring the ring requires only $O(1)$ rounds if the number of rounds is averaged over both the nodes and the identity assignments.
\end{abstract}


\thispagestyle{empty}
\newpage
\setcounter{page}{1}

\section{Introduction}

The $\LOCAL$ model \cite{Peleg00} is a standard model of distributed network computing. In this model, the network is abstracted as a graph, and the nodes perform in rounds to solve some task. At each round, each node can send messages to its neighbours in the graph, receive messages and perform some computation. The (time) complexity of an algorithm solving some task is measured by the number of rounds before the task is completed, which usually depends on the size of the network, that is, its number of nodes.

A classic assumption in the $\LOCAL$ model is that the nodes know the size of the network \emph{a priori}. As a consequence, in many algorithms, each node can compute from the start how many rounds are needed to solve the task, and stops after that number of rounds. 
There have been efforts to remove such \emph{a priori} knowledge about the graph, that is to avoid that the algorithm uses parameters such as the size of the graph, but also the arboricity \cite{BarenboimE10} or the maximum degree \cite{Musto11}). 
Quite recently a general technique, called \emph{pruning algorithms}, has been developed to remove the assumption that the nodes know the size $n$ of the network \cite{KormanSV13}. In other words, \cite{KormanSV13} provides a method to transform a non-uniform algorithm into a uniform algorithm.
The basic idea is to guess the number of nodes and to apply a non-uniform algorithm with this guess. 
The output can be incorrect, as the algorithm is only certified to be correct when it is given the actual number of nodes in the graph. 
The technique consists in virtually removing from the graph the nodes that have correct outputs, and to repeat the previous procedure with a new guess that is twice as large as the previous guess.
Eventually all nodes have an output after a certain number of iterations, and the solution that is computed is correct. 
Note that with the resulting uniform algorithm some nodes can output very quickly, and some others can output much later. 
So far, only the classic measure of complexity, \emph{i.e.} the time before all nodes stop and output, has been studied, even for algorithms with a such discrepancies in the running times. 
In other words, only the behaviour of the \emph{slowest} node has been considered. 
In this paper, we introduce a new measure of complexity, which is an \emph{average} measure, in opposition to the usual measure which is a \emph{worst-case} measure. 
More precisely, we define the running time of a node as the number of rounds before it outputs, and consider the average of the running times of the nodes. 
We argue that, when studying the locality of problems and of algorithms, it is worth to also consider this measure. 
Indeed it describes the typical local behaviour of the algorithm, that is, the behaviour of an \emph{ordinary} node.

In some contexts partial solutions are useful. For example, consider the scenario in which two tasks are to be performed one after the other. 
In such a case, it may happen that, on some part of the graph a partial solution for the first task is computed quickly. 
We can take advantage of this to start the second task in that part of the network, while the other nodes are still working on the first task. (Note that knowing if the first task is finished can be impossible locally, and one has to design the second algorithms such that it can start at different rounds on different nodes.) 
Consider a second scenario in which a global operator has to take a decision based on the outcome of a local algorithm. In that case, a partial solution may also be sufficient. 
For example the operator can detect that the network is in a bad state, and start immediately a recovery procedure without waiting for all nodes to finish.
Such situations are a motivation for the study of graph property testing, where a centralized algorithm probes the network on a sublinear number of nodes and take a decision based on this partial knowledge. We refer to the survey on graph property testing \cite{Goldreich10c} for more examples of applications.
When such partial solutions are useful, one would like to design algorithms that stop as soon as possible, and the average of the running times of the nodes is then a measure one would like to minimize.

Another classic assumption in the $\LOCAL$ model is that the nodes are given distinct identifiers. 
These identifiers (or IDs for short), are distinct binary strings on $O(\log n)$ bits, that is, distinct integers from a polynomially large space. 
The usual way to measure the complexity of an algorithm is again to consider the worst-case behaviour, that is, the performance of the algorithm on the worst ID assignment. 
We argue that the average performances over all ID assignments is also worth considering. 
Indeed many lower bounds are based on the fact that, as the identifiers can be viewed as set by an adversary, they do not really help to break symmetry. 
For example, on a path, one may consider the identifier assignment $1,2,...,n$, and argue that if the nodes only consider the relative ordering of the identifiers in their neighbourhoods, then many nodes have the same view, and thus they cannot break symmetry. 
It is interesting to study if such specific constructions are required, or if one can design lower bounds that are robust against arbitrary ID assignment. 
We cannot expect that IDs are always set in a perfect way for the task we consider, but it may seem excessive to consider that they are set in an adversarial way, which naturally leads to the question of random assignments. 
We study the complexity of algorithms on random ID assignment, as the average over all possible ID assignments of the running time of the slowest node. Finally, the typical behaviour of an algorithm can arguably be the expected running time of an ordinary node on a random ID assignment. That is, the standard complexity but averaged on both nodes and ID assignments. 

For the sake of concreteness, here is an example of the type of questions tackled in this paper. 
Consider the classic task of 3-colouring a ring of $n$ nodes. 
It is known that this task requires $\Omega(\log^*\!n)$ rounds \cite{Linial92}. 
This bound also holds for randomized algorithms \cite{Naor91}. The question tackled in this paper are of the following form: is it the case that a node typically outputs after a constant number of rounds, or is the $\Omega(\log^*\!n)$ lower bound robust to this spatial averaging? And what about the complexity of the problem on a random ID assignment?  

\paragraph*{Our results.}
Our first result is that averaging on the nodes can have a dramatic effect on the time complexity of solving a task in the $\LOCAL$ model. 
Indeed, for leader election on cycles, there is an exponential gap between the node-averaged complexity and the classic complexity. 
That is, the slowest node outputs after a number of rounds that is exponentially larger than the time complexity of an ordinary node.
This contrasts with our next result, for very sparse graphs. We say that a graph has linearly bounded growth, if there exists a constant $q$ such that every ball of radius $r$ has at most $q.r$ nodes. For such graphs, we show that, for many classic tasks, the two measures are of the same order of magnitude. 
More precisely for a class of tasks that generalizes the class of \emph{locally checkable labellings} ($\LCL$ for short) \cite{NaorS95}, we show the following lemma, that we call \emph{local average lemma}. 
For a given algorithm, either no node has running time much larger than the average running time in its neighbourhood, or there exists an algorithm that is strictly better. As a consequence when proving lower bounds for these problems, one can use the fact that, loosely speaking, there is no \emph{peak} in the distribution of the running times of the nodes. 
Then, to show that the average running time is large, it is sufficient to show that there exists a set of nodes with large running times, that are well spread out in the network.
This local average lemma can be used to show, for example, that for $\LCL$ problems on cycles, the landscape of complexities known for the slowest node (either $\Theta(1)$, $\Theta(\log^*\!n)$ or $\Theta(n)$) is the same for an ordinary node. 

We then move on to averaging on the identifier assignments. 
That is, we consider the expected behaviour of deterministic algorithms on random ID assignments. This topic happens to be related with the expected complexity of randomized algorithms. 
We show that even though these two models have specific properties, namely the independence of the random strings for randomized algorithms, and the uniqueness of the identifiers for random ID assignments, the complexities are essentially the same. It follows that the results known for randomized algorithms can be translated to random assignments.
 
Finally we prove that averaging on both nodes and IDs, can have an important impact on the complexity. We take the example of 3-colouring an $n$-node cycle. 
From the previous results of the paper, and from the literature, we know that this task has complexity $\Omega(\log^*\!n)$ for both the average on the nodes and the average on the identifiers. 
Quite surprisingly, when averaging on both the nodes and the ID assignments, the complexity becomes constant. In other words, deterministic and randomized complexity of ordinary nodes are clearly separated. Such separation contrasts with the case of the classic measure where randomized constant-time algorithms for $\LCL$ can be derandomized to get constant-time deterministic algorithms \cite{NaorS95}.

\paragraph*{Related works.}

The $\LOCAL$ model was defined formally in \cite{Linial92}, and a standard book on the topic is \cite{Peleg00}. 
The problem of leader election, studied in section \ref{sec:exponential_gap}, is a classic problem in distributed computing \cite{AttiyaW04, Lynch96}.

Deterministic algorithms stopping after different number of rounds on different nodes have been studied in contexts where the parameters of the graph, such as the degree or the number of vertices, are unknown. Such algorithms are called \emph{uniform algorithm}, because it is the same algorithm that is run on every graph, independently of the parameters. A work that is particularly relevant to us is \cite{KormanSV13}. In this paper the authors prove that for a wide class of problems, one can remove the assumption that the nodes know the size $n$ of the network. 
This is done by applying a general method to transform a non-uniform algorithm into a uniform one, without increasing of the asymptotic running time. 
In this framework, called \emph{pruning algorithms}, some nodes may stop very early and some may run for much longer time. Such algorithms justify the study of the behaviour of an ordinary node and not only of the behaviour of the slowest node.  

The local average lemma of section \ref{sec:local_average} applies to problems that are local in the sense that the nodes can check in constant time if a given solution is correct. This is an extension of the well-studied notion of locally checkable labelling (or $\LCL$ for short) \cite{NaorS95}. The original $\LCL$ requires in addition that the size of the inputs and outputs are bounded. Also the set of correct labellings usually studied, \emph{e.g.} in distributed decision \cite{FeuilloleyF16} or in $\LCL$, do not depend on the identifiers of the graph, a restriction that is not needed in the current paper.   

Randomized algorithms, that turn out to be equivalent to algorithms working on random ID assignments, form a well-studied subject, going back to the 80s with algorithms for finding a maximal independent set \cite{AlonBI86, Luby86}. Recently, improvements on classic problems have been obtained \cite{Ghaffari16, HarrisSS16} along with an exponential separation between randomized and deterministic complexity \cite{ChangKP16} (see also \cite{BrandtFHKLRSU16}). In \cite{Ghaffari16}, the author, by advocating the study of the so-called \emph{local complexity} for a randomized algorithms, conveys the same message as the current paper: the behaviour of a typical node is worth considering, even if some nodes of the graph have much worst behaviour.

In this paper, two relaxation of the classic measure are considered, from worst-case to average, on the nodes and on the IDs. An aspect that we do not consider is the structure of the graph. We refer to \cite{GamarnikS14} and references therein, for the topic of local algorithms on random graphs. 

\section{Model and definitions}\label{sec:model_def} 

\paragraph*{Graphs and neighbourhoods.} The graphs considered in this paper are simple connected graphs, and throughout the text $n$ will denote the number of nodes in the graph. The distance between two nodes is the number of edges on a shortest path between these nodes, that is, the hop-distance. The $k$-neighbourhood of a node $v$ in a graph $G$, is the graph induced by the nodes at distance at most $k$ from $v$. Every node is given a distinct identifier on $O(\log n)$ bits, or equivalently an integer from a polynomially large range.

\paragraph*{Distributed algorithms.} The algorithms studied in this paper can be defined in two ways. In both definitions, the nodes are synchronized and work in rounds, and for both the computational power of the nodes is unbounded. 
In the first definition, at each round, every node can exchange messages with its neighbours, and perform some computation. There is no bound on the size of the messages.
A given node chooses an output after some number of rounds, and different nodes can stop at different rounds.  
After the output, a node can continue to transmit messages and perform computations, but it cannot change its output. In other words, the nodes do not go to a sleep mode once they have output, but the output is irrevocable. 
In the second definition, each node starts with the information of its $0$-neighbourhood, and increases the size of this view at each round. That is, after $k$ rounds, it knows its $k$-neighbourhood. This $k$-neighborhood includes the structure of the graph in this neighbourhood, along with the identifiers and the inputs of each node. 
At some round, it chooses an output and stops.
These two definitions are equivalent. On one hand, if we start from the first definition, we can assume that each round every node sends to its neighbours all the information it has about the graph (remember that the message size is unbounded)%
\footnote{There is a subtlety here, which is that after $k$ rounds in the message-passing algorithm a node cannot know the edges that are between nodes at distance exactly $k$ from it. For the sake of simplicity, we consider the proper $k$-neighbourhoods, as it does not affect the asymptotics of the algorithms.}. 
Then after $k$ rounds, a node has gathered the information about its $k$-neighbourhood. On the other hand, given a $k$-neighbourhood, a node can simulate the run of the other nodes, and compute the messages that it would receive if the nodes were using a message-passing algorithm. 

\paragraph*{Complexity measures studied.} The running time of a node is the number of rounds before it outputs. With the second definition, the running time of the algorithm can be described in a more combinatorial way: it is the minimum integer $k$ such that the node can choose an (irrevocable) output based only on the view of radius $k$. 
Let the set of legal ID assignments be denoted by $\mathcal{ID}$. Given a graph $G$, an identifier assignment $I:v\mapsto \mathcal{ID}$, some input $x$, an algorithm $A$, and a node $v$, we denote by $r_{G,I,x,A}(v)$ the running time of node $v$ in this context.
When the context is clear, we simply use $r(v)$. We now define the different measures of complexity used in this paper. Given a graph $G$, and an algorithm $A$, we call \emph{complexity of the slowest node complexity} or \emph{classic complexity}, and \emph{complexity of an ordinary node} or \emph{node-averaged complexity} respectively, the following quantities:
\[
\max_{I\in \mathcal{ID}}\ \max_{v\in G}\ r_{G,I,A}(v) 
\text{\hspace{0.7cm} and \hspace{0.7cm}}
\max_{I \in \mathcal{ID}}\ \frac{1}{n} \sum_{v\in G} r_{G,I,A}(v).  
\]

In the second part of this paper, we consider the running time of the slowest node-averaged on the identifier assignments, and the running time averages on both the identifiers assignments and the nodes, that is, the following measures:

\[
\frac{1}{|\mathcal{ID}|} \sum_{I\in \mathcal{ID}}\ 
\left( \max_{v\in G} \ r_{G,I,A}(v) \right)
\text{\hspace{0.7cm} and \hspace{0.7cm}}
\frac{1}{|\mathcal{ID}|} \sum_{I\in \mathcal{ID}}\ 
\left( \frac{1}{n} \sum_{v\in G} r_{G,I,A}(v) \right).
\]

\paragraph*{Tasks and languages.} The tasks or problems that we want to solve in a distributed manner, are formalized with the notion of \emph{language}. A language $\mL$ is a set of configurations of the form $(G,I,x,y)$, where $G$ is a graph, $I$ an identifier assignment, and $x$ and $y$ are functions from the nodes of the graph to a set of labels. We are interested in constructing these languages, which means that given a graph $G$, an ID assignment $I$ and inputs given by the function $x$, we want to compute locally a function $y$ such that  $(G,I,x,y)$ is in the language $\mathcal{L}$. 
The languages considered are such that for every $(G,I,x)$, there exists a legal output $y$. Note that usually, the identifier assignment is not part of the language \cite{FeuilloleyF16, FraigniaudKP13, NaorS95}, but our results hold for this more general version.  

\paragraph*{Knowledge of the size of the network.} In section \ref{sec:exponential_gap}, we use the most general option regarding the knowledge of~$n$ by the nodes: we assume such knowledge for lower bounds, whereas for upper bounds we do not require it. For section \ref{sec:local_average}, we assume that nodes do not have the knowledge of $n$. For the randomized part we assume this knowledge for the sake of simplicity, and we refer to subsection 4.4 of \cite{KormanSV13} for a technique to remove such assumptions for randomized algorithms. 

\paragraph*{Additional notations.} Throughout the paper, the expression \emph{with high probability} means with probability at least $1-1/n$. Also, for a set $X$, $|X|$ denotes the cardinal of the set.

\section{Exponential gap for a global language}\label{sec:exponential_gap}

The complexity of an ordinary node is bounded by the complexity of the slowest node by definition. In this section, we show that the gap between these two quantities can be exponential. 
\begin{theorem}\label{thm:gap}
The gap between the node-averaged complexity and the classic complexity can be exponential.
\end{theorem}

We illustrate this phenomenon on the classic problem of leader election. The language of leader election is the set of graphs with arbitrary IDs, with no inputs and binary outputs, such that exactly one node has label 1, and the others have label 0.

\begin{proposition}[Folklore]\label{prop:leader_worst}
Leader election on an $n$-node ring requires $\Theta(n)$ rounds (for the slowest node).
\end{proposition}

This result is part of the folklore, but we prove this statement for completeness. The complexity of leader election in various models is discussed in \cite{AttiyaW04, Lynch96}.

\begin{proof}
Let $A$ be an algorithm for leader election, which has access to the size of the graph. Suppose that the slowest node complexity of $A$ is $c(n)\in o(n)$. Let $n_0$ be a large enough constant such that $2c(n_0)+1<n_0/2$. Consider a ring $R_1$ of length $n_0$. After running the algorithm $A$ on $R_1$, a node $v_1$ is elected to be the leader. This node $v_1$ outputs 1, after at most $c(n_0)$ steps. That is, $v_1$ outputs based on a view that contains at most $2c(n_0)+1$ nodes. 
Because of the definition of $n_0$, this view contains less than $n_0/2$ nodes. Let $I_1$ be the set of identifiers in this view. Now consider another ring $R_2$ of length $n_0$, whose set of identifiers does not contain any of the IDs of $I_1$. 
Again, a node $v_2$ is designated as the leader, and its view contains less than $n_0/2$ nodes. 
Now consider the ring made by concatenating the two views, and adding  dummy nodes with fresh identifiers, to make sure that the ring has size $n_0$. 
Because the identifiers are all distinct, this is a proper instance. 
Then, as $v_1$ and $v_2$ have the same view as in $R_1$ and $R_2$ respectively, with the same graph size $n_0$, they output the same as in $R_1$ and $R_2$ respectively. That is, they both output 1, and thus produce a configuration that is not in the language, which a contradiction.
\end{proof}

\begin{proposition}\label{prop:leader_average}
The complexity of an ordinary node for leader election on an $n$-node ring is $O(\log n)$.  
\end{proposition}

\begin{proof}
Consider the following algorithm. Each node increases its radius until one of the two following situations occurs. 
First, if it detects an ID that is larger than its own, it outputs 0. 
Second, if it can see the whole ring, and can detect no ID is larger than its own, then it outputs 1. It is easy to see that this algorithm is correct: exactly one node will output 1, the node with the largest ID. 
Note that this algorithm is order-invariant in the sense of \cite{NaorS95}, \emph{i.e.} the algorithm does not take into account the identifiers themselves, but only their relative ordering in its view. In particular, the algorithm does not require the knowledge of $n$. We show that the node-averaged complexity of this algorithm is logarithmic in $n$.  

Let us first make an observation. Consider the nodes with the $k$ largest identifiers, and mark them. 
The nodes that are not marked form $k$ paths. (Some of these paths can be empty, if two marked nodes are adjacent). A key property is that the behaviour of the algorithm on one path is independent of the other paths. 
More precisely, we claim that on a given path the algorithm will have the same behaviour whatever the sizes and the identifier distributions of the other paths are. Fix a path, and a node $v$, in this path. 
By definition, $v$ has an identifier that is smaller than the ones of the two marked nodes at the endpoints of the path. 
Therefore, either it stops before reaching an endpoint, or exactly when it reaches one of the marked nodes. 
As a consequence, such a node will output based only on its knowledge of the path. 
This simple observation implies that we can study the behaviour of the algorithm on each path separately. Let $p$ be an integer, and let us consider a path of length $p$ with two additional marked nodes at each endpoint. 
Thanks to order-invariance, it is sufficient to study the behaviour of the algorithm on this path with all the relative ordering of identifiers. 
Let $a(p)$ be the maximum over all these identifier assignments of the sum of the running times of the nodes. We claim that this function obeys the following recurrence relation:
\[
 a(p) = \max_{1 \leq k \leq \lceil p/2 \rceil}\left\{ k + a(k-1) + a(p-k) \right\}.
\]
Consider the node $v$ with the largest identifier in the path (excluding the marked endpoints). As noted before, it must reach one of the endpoints to stop. Then if we mark this node, the behaviour of the algorithm on the two subpath is independent of the context, and the maximum sums of running times in each path is $a(p_1)$ and $a(p_2)$ for the first subpath of length $p_1$ and the second of length $p_2$ respectively. Then the only parameter is the distance $k$ from $v$ to the closest endpoint. Given such an integer $k$, $a(p)$ is then equal to $k + a(k-1) + a(p-k)$. One can then check by induction that this maximum is always met for the value $k=\lceil p/2 \rceil$. Then an alternative formula is: 
\[
a(p)=\left\lceil \frac{p}{2} \right\rceil + a\left(\left\lceil \frac{p}{2}\right\rceil \right) + a\left(\left\lceil \frac{p}{2} \right\rceil - 1 \right).
\]
The sequence $a(n)$, defined by the induction formula above, along with initial values $a(0)=0$ and $a(1)=1$, is known to be in $\theta(n\log n)$. For references and more information about this sequence, see \cite{oeisA000788}.

When running the algorithm, the node with the largest identifier will see the whole graph and detect it has the largest ID, and then output 1. Its running time is then $\lceil n/2 \rceil$. We can then mark this node, and apply the result of the previous paragraphs to the path made by the remaining nodes. 
Consequently, the sum of the running times of the nodes is equal to $\lceil n/2 \rceil+a(n-1)$ which is in $\theta(n\log n)$. Thereafter, the complexity of an ordinary node is logarithmic in $n$. 
\end{proof}

Note that analysis of the same flavour already exist in the literature, see for example \cite{Santoro06} p.125. Theorem \ref{thm:gap} follows from propositions \ref{prop:leader_worst} and \ref{prop:leader_average}.

\section{Local average lemma and application}\label{sec:local_average}

This section is devoted to proving that, for local languages on very sparse graphs, the complexity of an ordinary node is basically the same as the one of the slowest node. This proof is based on a \emph{local average lemma}. 
Given a graph and an algorithm, let us define informally a \emph{peak}, as a node whose running time is much larger than the average running time in its neighbourhood at some distance. The lemma states that, for local languages, and for algorithm that are somehow optimal, there is no such peak. 

\subsection{Intuition on 3-colouring of a ring}\label{subsec:intuition}
In order to give an intuition of the lemma and its proof, and to justify the notions we introduce in the next paragraph, let us consider the example of 3-colouring a cycle. 
Consider an algorithm for the problem, and three adjacent nodes $u$, $v$ and $w$, in this order in a cycle. We claim that if $r(v)>\max(r(u),r(w))+1$, then the algorithm can be speeded up. Note that after $\max(r(u),r(w))+1$ steps, $v$ has a view that contains the whole views of $u$ and $w$. Then $u$ can simulate the computations of $u$ and $w$, and deduce the colours they output, and output a non-conflicting colour. 
As a consequence if one wants to prove a lower bound on the average of the running times, one can assume that $r(v)\leq \max(r(u),r(w))+1$. 
This leads to the fact there is no peak: every node with high running time has at least one neighbour with similar running time, and then the average running time at distance 1 cannot be smaller than half the running time of $v$. The lemma is a generalization of this observation, for further neighbours, more general graphs and more general problems.

\subsection{Additional definitions}
In order to state the lemma we need to introduce a few notions.

\paragraph*{Class $\LCLS$.}
We consider a large class of distributed problems that we call $\LCLS$, which includes the well-known class of $\LCL$ problems \cite{NaorS95}, and the more general class $\LD$ \cite{FraigniaudKP13}.
A language $\mL$ is in $\LCLS$, if there exists a constant-time \emph{verification algorithm}, that is an algorithm $\mV$ performing in a constant number of rounds, with binary output, \emph{accept} or \emph{reject}, such that the following holds. For every configuration $(G,I,x,y)$, $\mV$ accepts at every node, if and only if the graph is in the language $\mL$. The running time of $\mV$ is called the \emph{verification} radius.
No bound on the size of the inputs and output is necessary, and the language can depend on the identifiers.

\paragraph*{Graphs with linearly bounded growth.} A graph has \emph{linearly bounded growth} if there exists a constant $q$ such that if any ball of radius $r$ contains at most $qr$ nodes. The constant $q$ is called the growth parameter. For example a cycle has linearly bounded growth, with parameter $3$.

\paragraph*{Minimal algorithms.}
We would like to write a statement of the following form: given a node $v$ whose running time is $r$, the nodes of its neighbourhood have running times whose average is roughly $r$. This type of statement cannot hold in general as we could artificially increase the radius of a node by modifying the algorithm. But as we are interested in lower bounds we can consider algorithms that are in some sense optimal. 
More precisely, let $A$ and $A'$ be two distributed algorithms for some language $\mL$. We say that $A$ is smaller than $A'$, if on every graph, every ID assignment and inputs, and on every node, the running time of $A$ is at most the running time of $A'$. For lower bounds on the node-averaged complexity, it is sufficient to study algorithms that are minimal for this ordering. Indeed, if an algorithm that is not minimal has low complexity, then there exists a minimal algorithms that has at most this complexity.

\paragraph*{Knowledge of $n$.} In this section the algorithm do not have the knowledge of $n$.

\paragraph*{Additional notations.}
Let us denote by $B(v,k,G,I,x)$ the subgraph of $G$, with identifiers~$I$, and inputs~$x$, induced by the nodes at distance at most~$k$ from a node~$v$.
Likewise, given two integers $k_1<k_2$, let $S(v,k_1,k_2,G,I,x)$ be the induced graphs with IDs and inputs, induced by the set of nodes whose distance to~$v$ is at least $k_1$ and at most $k_2$. Such set of nodes are referred to as \emph{crowns} in the following. When the context is unambiguous, we may omit $G$, $I$ and $x$. 

\subsection{Lemma statement}
\begin{lemma}[Local average lemma]\label{lem:local_average}
Let $\mL$ be a language in $\LCLS$, and let $\mathcal{F}$ be a graph family with linearly bounded growth and $A$ be a minimal algorithm for $\mL$. There exists two positive constants $\alpha$ and $\beta$, such that for any graph of $\mathcal{F}$, ID assignment, inputs, and node $v$, the average of the running time of $A$ on the nodes at distance at most $r(v)/2$ from $v$, is at least $\alpha.r(v)-\beta$.
\end{lemma}

\subsection{Proof of the lemma}

Let $\mL$, $A$, $G$, $I$, $v$ and $x$ be respectively, a language, a minimal algorithm, a graph, an ID assignment, a node and an input assignment as in the lemma.
In this proof, several graphs, inputs, IDs and algorithms are considered ; when not specified, we refer to the elements we have just defined. For example $r(v)$ refers to the running time of $A$ on $v$ in G, with $I$ and $x$. 
Let $\mV$ be the verification algorithm of $\mL$, and let $t$ be the verification radius of $\mV$. Let $q$ be the growth parameter of $\mathcal{F}$.

The proof is in two steps, that we highlight with two technical claims. The first claim relates the running time of a node with the running time of the nodes in a crown around it. The proof uses a simulation argument as in the example of 3-colouring in subsection \ref{subsec:intuition}, and we call it the \emph{simulation step} in the following. 

\begin{claim}[Simulation step]\label{clm:rv_bound}
For every integer $k$:
$$r(v) \leq 2k+2t + \max_{u\in S(v,k,k+2t)}r(u).$$ 
\end{claim}

\begin{proof}[Proof (Simulation step)]
For the sake of contradiction, suppose the inequality does not hold for some fixed~$k$.
 Let us use the following notations: 
\[
M=\max_{u\in S(v,k,k+2t)}r(u) 
\text{\hspace{0.7cm} and \hspace{0.7cm}}
B=B(v,k+2t+M).
\]
As in the simulation of subsection \ref{subsec:intuition}, we show how to craft a new algorithm $A'$, smaller than $A$. 

\paragraph*{Definition of $A'$.} Consider a node $w$ of a graph $H$, with ID assignment $I_H$, and inputs $x_H$. The behaviour of~$A'$ on this graph differs from the behaviour of $A$ only if the following conditions are fulfilled (see figure \ref{apx:fig:local_average}): 
\begin{enumerate}
\item[(1)] The running time of $A$ on $w$ in $(H, I_H, x_H)$, is at least $2k+2t+M$ ;
\item[(2)] The node $w$ is at distance at most $k$ from a node $v_H$ whose neighbourhood at distance $k+2t+M$ is exactly  $B$. 

\end{enumerate}

\begin{figure}[h!]
\begin{center}
\scalebox{0.8}{
\begin{tabular}{cc}
\scalebox{1.5} 
{
\begin{pspicture}(0,-1.72)(4.82,1.72)
\definecolor{bleu1}{rgb}{0.35,0.35,0.7}
\definecolor{bleu2}{rgb}{0.5,0.5,0.85}
\definecolor{bleu3}{rgb}{0.7,0.7,1.0}
\definecolor{jaune}{rgb}{1.0, 0.90, 0.6}
\definecolor{orange}{rgb}{1.0, 0.75, 0.6}
\pspolygon[linearc=0.3, linewidth=0.04,fillstyle=solid, fillcolor=jaune, linecolor=orange](0.0,1.8)(3.6,1.8)(4.8,0.7)(4.8,-1.8)(0.0,-1.8)
\pscircle*[linecolor=bleu3](2.9,0.0){1.5}
\pscircle*[linecolor=bleu2](2.9,0.0){0.9}
\pscircle*[linecolor=bleu1](2.9,0.0){0.5}

\psline[linewidth=0.02cm,arrowsize=0.05291667cm 2.0,arrowlength=1.4,arrowinset=0.4]{<->}(1.86,-1.06)(2.28,-0.64)
\psline[linewidth=0.02cm,arrowsize=0.05291667cm 2.0,arrowlength=1.4,arrowinset=0.4]{<->}(2.16,-0.48)(2.52,-0.24)
\psline[linewidth=0.02cm,arrowsize=0.05291667cm 2.0,arrowlength=1.4,arrowinset=0.4]{<->}(2.4,0.0)(2.85,0.0)

\psdots[dotsize=0.12](2.9,0.0)
\usefont{T1}{ptm}{m}{n}
\rput(1.8,0.7){{\small $B$}}
\rput(0.3,1.4){{\small $G$}}
\rput(1.9,-0.7){{\footnotesize $M$}}
\rput(2.2,-0.1){{\footnotesize $2t$}}
\rput(2.6,0.2){{\footnotesize $k$}}
\rput(3.0,0.1){{\footnotesize $v$}}
\end{pspicture} 
}
&
\scalebox{1.5} 
{
\begin{pspicture}(0,-1.72)(4.82,1.72)
\definecolor{bleu1}{rgb}{0.35,0.35,0.7}
\definecolor{bleu2}{rgb}{0.5,0.5,0.85}
\definecolor{bleu3}{rgb}{0.7,0.7,1.0}
\definecolor{jaune}{rgb}{1.0, 0.90, 0.6}
\definecolor{orange}{rgb}{1.0, 0.75, 0.6}

\pspolygon[linearc=0.3, linewidth=0.04,fillstyle=solid, fillcolor=jaune, linecolor=orange](0.0,1.81)(4.88,1.81)(4.84,-0.79)(3.84,-1.81)(0.02,-1.81)
\pscircle*[linecolor=orange](1.91,0.0){1.75}
\pscircle*[linecolor=bleu3](1.76,-0.03){1.5}
\pscircle*[linecolor=bleu2](1.76,-0.03){0.9}
\pscircle*[linecolor=bleu1](1.76,-0.03){0.5}
\psdots[dotsize=0.12](1.76,-0.03)
\psdots[dotsize=0.12](2.0,-0.03)

\psline[linewidth=0.02cm,arrowsize=0.05291667cm 2.0,arrowlength=1.4,arrowinset=0.4]{<->}(2.05,-0.03)(3.65,-0.03)

\usefont{T1}{ptm}{m}{n}
\rput(0.7,0.7){{\small $B$}}
\rput(0.3,1.4){{\small $H$}}
\rput(1.6,0.15){{\footnotesize $v_H$}}
\rput(2.0,0.15){{\footnotesize $w$}}
\rput(3.0,0.5){{\footnotesize $2k+2t$}}
\rput(3.0,0.15){{\footnotesize $+M$}}

\end{pspicture} 
}
\end{tabular}
}
\end{center}
\caption{\label{apx:fig:local_average}
This figure illustrates the definition of the algorithm $A'$ in the proof of lemma \ref{lem:local_average}. On the left is the original graph $G$ with node $v$, along with the ball $B$ around it. The behaviour algorithm $A'$ differs from the algorithm $A$ only if it is in the situation of the node $w$ on the right: it has running time at least $2k+2t+M$, and it is at distance at most $k$ from a node whose $(k+2t+M)$-neighbourhood is exactly $B$.
}
\end{figure}
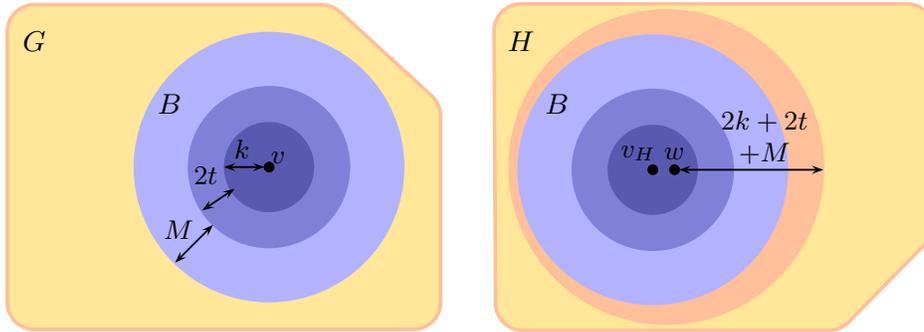 

When the two conditions are fulfilled, let $w_G$ be the node of $G$, whose position in $B$, ID, and input, are the same as the ones of $w$ in $H$. In this precise case, the algorithm $A'$ stops at round $2k+2t+M$, and outputs the same label as $A$ does on $w_G$, in $(G, I, x)$. Loosely speaking, if a node detects that it is in the core of a ball identical $B$, then it simulates the algorithm $A$ on $G$ on the adequate node, and outputs the same label. (Note that because the running time is $2k+2t+M$, and because $k$ is at distance at most $k$ from $v_H$, it can see whether condition (2) is fulfilled or not.)

\paragraph*{Correctness of $A'$.}
Consider an arbitrary graph $(H,I_H,x_H)$. Remember that $A$ is a correct algorithm for~$\mL$. Note that by definition the output of $A'$ may differ from the one of $A$ only if $H$ contains $B$, thus $A'$ is correct on every graph not containing $B$. 
Then if $H$ contains $B$, only the nodes at distance at most $k$ from~$v_H$, the center of the ball $B$, may have changed their outputs. 
To decide if these outputs are correct for the language~$\mL$, we use the verification algorithm~$\mV$. Remember that this algorithm has constant verification radius $t$. 
As only the inputs at distance at most $k$ from $v_H$ are modified, the nodes where $\mV$ may reject are the ones in the ball of radius $k+t$ (centred on $v_H$). The view of the verification algorithm on such a nodes are then included in the ball of radius $k+2t$. 

We claim that the outputs of $A'$ in the ball of radius $k+2t$ around $v_H$ in $H$,  correspond to the outputs of $A$ in the ball of radius $k+2t$ around $v$ in $G$. This claim implies that the verification algorithm will accept on the nodes of the ball of radius $k+t$. Indeed, the verification algorithm will have the same view as on $G$ with the outputs of $A$, and as $A$ is correct, it accepts. This in turn implies that $A'$ is correct.

Let us prove the claim of the previous paragraph. The nodes at distance $k$ from $v_H$ have by definition the same outputs as the corresponding nodes in $G$. The nodes in the crown $S(v,k+1,k+2t,H,I_H,x_H)$ have the same view in $(H,I_H,x_H)$ as in $(G,I,x)$, because these graphs coincide on $B$, and no node has running time large enough to see something outside of $B$. Indeed these nodes are at distance at least $M$ from the boundary of $B$, and by definition $M=\max_{u\in S(v,k,k+2t)}r(u)$.

\paragraph*{$A'$ is strictly smaller than $A$.}
The algorithm $A'$ is the same as $A$ except on the nodes which fulfil the two conditions  at the beginning of this proof. The nodes that have been modified had running time at least $2k+2t+M$ with $A$ and have running time exactly $2k+2t+M$ with $A'$. Thus $A'$ is smaller than $A$, and it is strictly smaller because we assumed that $v$ had running time strictly larger that $2k+2t+M$ in $A$, and it has running time exactly $2k+2t+M$ in $A'$.

Finally, $A'$is a correct algorithm, strictly smaller that $A'$, thus $A$ is not minimal. This is a contradiction. 
\end{proof}

Note that the hypothesis $\mL \in \LCLS$ is crucial in the later proof: it is because the correctness of the output is evaluated locally that it is safe to change some outputs, checking that these new outputs locally match the rules of the language.

We now move to the second part of the proof of lemma \ref{lem:local_average}, which is proving a second technical claim. 

\begin{claim}[Summation]There exists two constants $\alpha$ and $\beta$ such that:
\[ 
\alpha.r(v)-\beta \leq \frac{1}{|S(v,1,r(v)/2)|}\sum_{u\in S(v,1,r(v)/2)}r(u). 
\]
\end{claim}

To make the proof look more natural, we again give an intuition on an easier case. Consider a simplified version of the inequality of the simulation claim: $r(v)\leq \max_{u\in S_k}r(u)$, where~$S_k$ is the set of nodes at distance exactly~$k$. The quantity $\max_{u\in S_k}r(u)$ is upper bounded by $\sum_{u\in S_k}r(u)$. 
Then, summing both terms of the inequality, for $k$ ranging from 1 to $r(v)$, one gets $r(v)^2\leq \sum_{u\in S}r(u)$, where $S$ is the ball of radius~$r(v)$, without~$v$. 
Now the bounded growth helps us to bound the cardinal of $S$. Namely, as the growth parameter is $q$, then there are at most~$q\cdot r(v)$ nodes in~$S$. Then $\sum_{u\in S}r(u)\leq q\cdot r(v) \cdot a_S$, where $a_S$ is the average running time in~$S$. Then $r(v)^2 \leq q\cdot a_S \cdot r(v)$, thus $\frac{1}{q} r(v) \leq a_S$. This corresponds to a simplified form of the summation claim, with $\alpha=\frac{1}{q}$, and $\beta=0$.

\begin{proof}[Proof (Summation)] 
Claim \ref{clm:rv_bound} states that for every $k$,
\[
r(v) \leq 2k+2t+\max_{u\in S(v,k,k+2t)}r(u), 
\]
then, 
\[
r(v)-2k-2t \leq \sum_{u\in S(v,k,k+2t)} r(u).
\]	
Let us sum the inequality above, for $k$ ranging from 1 to $r(v)/2-2t$. We assume without loss of generality that $t$ and $r(v)$ are positive. The sum of the left-hand terms is:
\[
\sum_{k=1}^{r(v)/2-2t}\left(r(v)-2k-2t\right)
=\frac{r(v)^2}{4}-t\cdot r(v)+\frac{r(v)}{2}-2t
\geq \frac{r(v)^2}{4}-3t\cdot r(v).
\]
The sum for the right-hand terms is:
\[ 
\sum_{k=1}^{r(v)/2-2t}\sum_{u\in S(v,k,k+2t)} r(u)\leq (2t+1)\times\sum_{u\in S(v,1,r(v)/2)}r(u).
\]
This is because the radius of a fixed node appears at most $2t+1$ times in the sum, because it is part of at most $2t+1$ crowns of the form $S(v,k,k+2t)$.
Then because of the bounded growth, the number of nodes in $S(v,1,r(v)/2)$ is bounded by $q\cdot r(v)/2$. Then the following holds:
\[ 
(2t+1)\times\sum_{u\in S(v,1,r(v)/2)}r(u) \leq (2t+1)\times \frac{q\cdot r(v)}{2}\times \frac{1}{|S(v,1,r(v)/2)|}\times \sum_{u\in S(v,1,r(v)/2)}r(u).
\]
Now using the lower bound of the left-hand term and the upper bound of the right hand term, and using $(2t+1)/2 \leq 2t$, we get: 
\[
\frac{r(v)^2}{4}-3t\cdot r(v) \leq 2t\cdot q\cdot r(v)\frac{1}{|S(v,1,r(v)/2)|}\sum_{u\in S(v,1,r(v)/2)}r(u).
\]
Dividing by $2t\cdot q\cdot r(v)$ on both sides, and defining $\alpha=\frac{1}{8t\cdot q}$ and $\beta=\frac{3}{2q}$, we get:
\[ 
\alpha.r(v)-\beta \leq \frac{1}{|S(v,1,r(v)/2)|}\sum_{u\in S(v,1,r(v)/2)}r(u) 
\]
\end{proof}

As the right-hand term in the inequality of the second claim is the average running time in the $r(v)/2$ neighbourhood of $v$, this concludes the proof of lemma \ref{lem:local_average}.

\subsection{Applications}

Thanks to the lemma, establishing a lower bound for node-averaged complexity of languages in $\LCLS$ for very sparse graphs boils down to show a simpler fact. 
It is sufficient to prove that there exists a set of nodes spread across the network with large enough running times. Then, as the nodes in the neighbourhood of these nodes have  similar running times in average, we get a large average for the whole network.  
We illustrate this type of proof with $\LCL$ problems on cycles. It is known that for such problems, the slowest node complexity can only take three forms: $O(1)$, $\Theta(\log^*\!n)$ or $\Theta(n)$. See for example \cite{BrandtHKLOPRSU17} for a recent presentation of this classification.\footnote{Even if not stated explicitly in \cite{BrandtHKLOPRSU17}, this classification also holds in the context where no knowledge of $n$ is assumed. This is because the $\Theta(\log^*\!n)$ bound relies on the construction of a maximal independent set, and that MIS is a problem for which the construction of \cite{KormanSV13} works.} We prove that the situation is exactly the same for ordinary nodes.

\begin{theorem}\label{thm:LCL_average}
For $\LCL$ on cycles, the node-averaged complexity has the same asymptotic classification as the slowest node complexity. 
\end{theorem}

\begin{proof}
Remember that the slowest node complexity is an upper bound on the node-averaged complexity. Thereafter, it is sufficient to only prove the two lower bounds: $\Omega(\log^*\!n)$ and $\Omega(n)$.

Let us first focus on the case $\Theta(n)$. In this case, there exists a constant $\gamma$ (with $0<\gamma\leq 1$), such that on every cycle on $n$ nodes, for large enough $n$, at least one node $v$ has a running time at least $\gamma n$. 
As we consider a lower bound, we can assume that the algorithm is minimal.  
As cycles have linearly bounded growth, lemma \ref{lem:local_average} applies, the average complexity in the $(\gamma\cdot n/2)$-neighbourhood of $v$ is at least $\alpha \gamma\cdot n - \beta$, where $\alpha$ and $\beta$ are constants.
Thereafter, the sum of the running times, in the $(\gamma\cdot n/2)$-neighbourhood of $v$ is bounded from below by $\alpha\gamma^2n^2 - \beta \gamma n$. 
Hence, dividing by the number of nodes, the average complexity for the whole cycle is in $\Omega(n)$. 

Let us now consider the case of classic complexity $\Theta(\log^*\!n)$. Consider any minimal algorithm $A$ for the language $\mL$ we consider. 
Again, let $\gamma$ be a constant, such that the slowest node complexity is at least $\gamma\log^*\!n$, for large enough~$n$.
Let $R_1$ be a ring on $n$ nodes, such that a node $v_1$ has running time $r_1\geq \gamma \log^*\!n$. 
Then let $H_1$ be the graph that is composed of only the $r_1$-neighbourhood of $v_1$, and let $I_1$ be the set of identifiers of this segment. 
Now consider another ring $R_2$ on $n$ nodes, with no identifiers from $I_1$, such that there exists a node $v_2$ with running time $r_2\geq \gamma\log^*\!n$.
Let $H_2$ be the concatenation of $H_1$ with the $r_1$-neighbourhood of $v_1$. Note that because no identifier from $I_1$ is present in $R_2$, $H_2$ has distinct identifiers.
This operation can be repeated, until $H_k$ has more than $n/2$ nodes. Let $H$ be $H_k$, completed in an arbitrary way to get a full cycle of size $n$ with distinct identifiers. 

Note that as we performed the operation at most a linear number of times, the fact of removing some identifiers at each step is harmless as the identifier space is supposed to be polynomially large. 
Also note that the $\Theta(\log^*\!n)$ lower bound for the classic complexity is not affected by the constraints we add on the identifier space. This is because the lower bounds proofs do not rely on the particular shape of this space, which can even be assumed to be $\{1...n\}$\cite{Linial92}.

We claim that on this cycle $H$ with this ID assignment, the node-averaged complexity is $\delta\log^*\!n$ for some constant~$\delta$. 
Indeed the nodes $(v_i)_i$, for $i$ ranging from 1 to $k$, have the same neighbourhoods as in $(R_i)_i$ respectively, thus have running times $(r_i)_i$ respectively. 
Then using lemma \ref{lem:local_average}, for every $i$, the nodes at distance at most $r_i/2$ from $v_i$ have running time at least $\alpha r_i-\beta$. 
And by construction there is a constant fraction of the nodes of $H$ that are in this case.
As for every $i$, $r_i\geq \gamma\log^*\!n$, a constant fraction of the nodes have running time at least $\alpha \gamma\log^*\!n-\beta$, which gives an average lower bounded by $\delta\log^*(n)$, for some~$\delta$.
\end{proof}

This ``extract and glue'' technique works on other classes of graphs, and similar bounds can thus be achieved. 
Nevertheless it is not true that, for any $\LCL$ problem and any graph class, the classic complexity is the same as the node-averaged complexity, as the following proposition shows. 

\begin{proposition}\label{prop:special_class}
There exists a graph class $\mathcal{C}$ for which $3$-colouring has slowest node complexity $\Omega(\log^*\!n)$ but node-averaged complexity in $O(1)$.
\end{proposition}

\begin{proof}
Consider the following construction, illustrated by figure~\ref{fig:caterpillar}. 
Start with a path of even length $k$, and index the nodes along the path from $v_1$ to $v_k$. 
Create three new nodes and link them to the node $v_k$. Now for the nodes $v_i$ with $1<i<k$, if the index $i$ is even, then add a node $v_i'$ and the edge $(v_i,v_i')$. We call this construction a \emph{short leg}. 
If the index $i$ is odd, add two nodes $v_i'$ and $v_i''$, and two edges $(v_i,v_i')$ and $(v_i',v_i'')$. This is a \emph{long leg}. 
For both types, the node $v_i$ is called the \emph{basis} of the leg. 
Let us call such a graph an even-odd caterpillar. 
The class $\mathcal{C}$ we consider is the set of graphs that can be built the following way: take an even-odd caterpillar based on a path of length $k$, and a cycle of length $\alpha\log^*\!k$ (where $\alpha$ is a large enough constant), and add an edge between an arbitrary node of the cycle and~$v_1$. 

Every algorithm 3-colouring this graph must in particular colour the $\alpha\log^*\!k$ cycle, and as the  size of the graph is linear in $k$, the identifiers space is polynomial in~$k$. 
Then Linial's lower bound applies on the cycle, and the slowest node complexity is $\Omega(\log^*\!k)$\footnote{A subtlety is that the cycle has one special node: the one on which the caterpillar is rooted. Linial's bound still holds because there exists nodes that are far enough from this special nodes, because $\alpha$ is chosen to be large enough.}. 

Let us now show that there exists an algorithm with constant node-averaged complexity for 3-colouring the graphs of this class. 
Every node first gathers its 3-hop neighbourhood. 
From this view it can deduce its position in the graph, and its behaviour for the following steps. 
More precisely, for every node~$v$: 
\begin{itemize}
\item if all the nodes that are adjacent to $v$ have degree two, then it is a node of the cycle, then it runs the Cole-Vishkin procedure for 3-colouring a cycle \cite{ColeV86}. It does not take into account the rest of the graph;
\item if it is the basis of a short leg, or the middle of a long leg, then it takes colour 1;
\item if it is the basis of a long leg, or has degree 1, then it takes colour~2;
\item if it has degree four, then it is $v_k$ and it takes colour $1$;
\item if it has degree two and both its neighbours have degree three, then it is~$v_1$, and it waits until both its neighbours have output, and it outputs a non-conflicting colour.
\end{itemize}

See figure \ref{fig:caterpillar}.
\begin{figure}[h!]
\begin{center}
\scalebox{1.2} 
{
\begin{pspicture}(0,-1.8)(10.62,1.8)
\definecolor{bleu1}{rgb}{0.35,0.35,0.7}
\definecolor{bleu2}{rgb}{0.5,0.5,0.85}
\definecolor{bleu3}{rgb}{0.7,0.7,1.0}
\definecolor{rouge1}{rgb}{1.0,0.7,0.7}
\definecolor{rouge2}{rgb}{1.0,0.4,0.4}
\definecolor{jaune}{rgb}{1.0, 0.90, 0.6}
\definecolor{jaune2}{rgb}{1.0, 0.6, 0.6}

\psline[linewidth=0.04cm](2.9,0.08)(10.1,0.08)
\pscircle[linewidth=0.04, dimen=outer, fillstyle=solid, fillcolor=rouge1, linecolor=rouge2](3.5,0.08){0.13}
\pscircle[linewidth=0.04,dimen=outer](1.22,0.02){1.7}
\pscircle[linewidth=0.04, dimen=outer, fillstyle=solid, fillcolor=jaune, linecolor=jaune2](2.89,0.08){0.13}
\pscircle[linewidth=0.04,dimen=outer,fillstyle=solid, fillcolor=jaune, linecolor=jaune2](-0.47,0.09){0.13}
\pscircle[linewidth=0.04, dimen=outer, fillstyle=solid, fillcolor=bleu3, linecolor=bleu1](2.75,0.69){0.13}
\pscircle[linewidth=0.04,dimen=outer,fillstyle=solid, fillcolor=rouge1, linecolor=rouge2](1.27,1.7){0.13}
\pscircle[linewidth=0.04,dimen=outer, fillstyle=solid, fillcolor=rouge1, linecolor=rouge2](2.43,1.17){0.13}
\pscircle[linewidth=0.04, dimen=outer, fillstyle=solid, fillcolor=jaune, linecolor=jaune2](1.95,1.51){0.13}
\pscircle[linewidth=0.04, dimen=outer, fillstyle=solid, fillcolor=jaune, linecolor=jaune2](0.01,1.19){0.13}
\pscircle[linewidth=0.04,dimen=outer, fillstyle=solid, fillcolor=jaune, linecolor=jaune2](-0.07,-1.07){0.13}
\pscircle[linewidth=0.04,dimen=outer, fillstyle=solid, fillcolor=bleu3, linecolor=bleu1](1.19,-1.67){0.13}
\pscircle[linewidth=0.04,dimen=outer, fillstyle=solid, fillcolor=bleu3, linecolor=bleu1](2.52,-1.07){0.13}
\pscircle[linewidth=0.04,dimen=outer, fillstyle=solid, fillcolor=bleu3, linecolor=bleu1](0.59,1.57){0.13}
\pscircle[linewidth=0.04,dimen=outer,fillstyle=solid, fillcolor=bleu3, linecolor=bleu1](-0.32,0.69){0.13}
\pscircle[linewidth=0.04,dimen=outer, fillstyle=solid, fillcolor=rouge1, linecolor=rouge2](-0.39,-0.53){0.13}
\pscircle[linewidth=0.04,dimen=outer,fillstyle=solid, fillcolor=rouge1, linecolor=rouge2](0.47,-1.49){0.13}
\pscircle[linewidth=0.04,dimen=outer,fillstyle=solid, fillcolor=rouge1, linecolor=rouge2](1.95,-1.49){0.13}
\pscircle[linewidth=0.04,dimen=outer,fillstyle=solid, fillcolor=rouge1, linecolor=rouge2](2.81,-0.55){0.13}

\psline[linewidth=0.04cm](4.12,0.12)(4.12,0.5)
\pscircle[linewidth=0.04,dimen=outer,fillstyle=solid, fillcolor=rouge1, linecolor=rouge2](4.13,0.63){0.13}
\pscircle[linewidth=0.04,dimen=outer,fillstyle=solid, fillcolor=bleu3, linecolor=bleu1](4.11,0.08){0.13}

\psline[linewidth=0.04cm](10.1,0.08)(10.34,0.48)
\psline[linewidth=0.04cm](10.1,0.08)(10.58,0.08)
\psline[linewidth=0.04cm](10.1,0.1)(10.3,-0.28)
\pscircle[linewidth=0.04,dimen=outer,fillstyle=solid, fillcolor=rouge1, linecolor=rouge2](10.39,0.57){0.13}
\pscircle[linewidth=0.04,dimen=outer, fillstyle=solid, fillcolor=rouge1, linecolor=rouge2](10.69,0.08){0.13}
\pscircle[linewidth=0.04,dimen=outer,fillstyle=solid, fillcolor=rouge1, linecolor=rouge2](10.39,-0.37){0.13}
\pscircle[linewidth=0.04,dimen=outer,fillstyle=solid, fillcolor=bleu3, linecolor=bleu1](10.1,0.08){0.13}

\psline[linewidth=0.04cm](4.72,0.12)(4.72,1.08)
\pscircle[linewidth=0.04,dimen=outer,fillstyle=solid, fillcolor=rouge1, linecolor=rouge2](4.73,1.19){0.13}
\pscircle[linewidth=0.04,dimen=outer,fillstyle=solid, fillcolor=bleu3, linecolor=bleu1](4.73,0.63){0.13}
\pscircle[linewidth=0.04,dimen=outer,fillstyle=solid, fillcolor=rouge1, linecolor=rouge2](4.73,0.08){0.13}

\psline[linewidth=0.04cm](5.34,0.1)(5.34,0.48)
\pscircle[linewidth=0.04,dimen=outer, fillstyle=solid, fillcolor=rouge1, linecolor=rouge2](5.35,0.61){0.13}
\pscircle[linewidth=0.04,dimen=outer, fillstyle=solid, fillcolor=bleu3, linecolor=bleu1](5.33,0.08){0.13}

\psline[linewidth=0.04cm](5.94,0.1)(5.94,1.06)
\pscircle[linewidth=0.04,dimen=outer, fillstyle=solid, fillcolor=rouge1, linecolor=rouge2](5.95,1.17){0.13}
\pscircle[linewidth=0.04,dimen=outer, fillstyle=solid, fillcolor=bleu3, linecolor=bleu1](5.95,0.61){0.13}
\pscircle[linewidth=0.04,dimen=outer, fillstyle=solid, fillcolor=rouge1, linecolor=rouge2](5.95,0.08){0.13}

\psline[linewidth=0.04cm](6.54,0.08)(6.54,0.46)
\pscircle[linewidth=0.04,dimen=outer, fillstyle=solid, fillcolor=rouge1, linecolor=rouge2](6.55,0.59){0.13}
\pscircle[linewidth=0.04,dimen=outer, fillstyle=solid, fillcolor=bleu3, linecolor=bleu1](6.53,0.08){0.13}

\psline[linewidth=0.04cm](7.14,0.08)(7.14,1.04)
\pscircle[linewidth=0.04,dimen=outer, fillstyle=solid, fillcolor=rouge1, linecolor=rouge2](7.15,1.15){0.13}
\pscircle[linewidth=0.04,dimen=outer, fillstyle=solid, fillcolor=bleu3, linecolor=bleu1](7.15,0.59){0.13}
\pscircle[linewidth=0.04,dimen=outer, fillstyle=solid, fillcolor=rouge1, linecolor=rouge2](7.15,0.08){0.13}

\psline[linewidth=0.04cm](7.74,0.06)(7.74,0.46)
\pscircle[linewidth=0.04,dimen=outer, fillstyle=solid, fillcolor=rouge1, linecolor=rouge2](7.75,0.57){0.13}
\pscircle[linewidth=0.04,dimen=outer, fillstyle=solid, fillcolor=bleu3, linecolor=bleu1](7.73,0.08){0.13}

\psline[linewidth=0.04cm](8.32,0.08)(8.32,1.04)
\pscircle[linewidth=0.04,dimen=outer, fillstyle=solid, fillcolor=rouge1, linecolor=rouge2](8.33,1.15){0.13}
\pscircle[linewidth=0.04,dimen=outer, fillstyle=solid, fillcolor=bleu3, linecolor=bleu1](8.33,0.59){0.13}
\pscircle[linewidth=0.04,dimen=outer, fillstyle=solid, fillcolor=rouge1, linecolor=rouge2](8.33,0.08){0.13}

\psline[linewidth=0.04cm](8.94,0.06)(8.94,0.46)
\pscircle[linewidth=0.04,dimen=outer, fillstyle=solid, fillcolor=rouge1, linecolor=rouge2](8.95,0.57){0.13}
\pscircle[linewidth=0.04,dimen=outer, fillstyle=solid, fillcolor=bleu3, linecolor=bleu1](8.93,0.08){0.13}

\psline[linewidth=0.04cm](9.54,0.06)(9.54,1.02)
\pscircle[linewidth=0.04,dimen=outer, fillstyle=solid, fillcolor=rouge1, linecolor=rouge2](9.55,1.13){0.13}
\pscircle[linewidth=0.04,dimen=outer, fillstyle=solid, fillcolor=bleu3, linecolor=bleu1](9.55,0.57){0.13}
\pscircle[linewidth=0.04,dimen=outer, fillstyle=solid, fillcolor=rouge1, linecolor=rouge2](9.55,0.08){0.13}
\end{pspicture} 
}
\end{center}
\caption{\label{fig:caterpillar}
The figure illustrates proof of proposition \ref{prop:special_class}. It takes $O(\log^*\!n)$ rounds to 3-colour the cycle on the left, but it take constant time to colour the even-odd caterpillar on the right, as a 2-colouring is hard-coded in the structure of the graph. In this picture, colour 1 is blue, colour 2 is red, and colour 3 is yellow.} 
\end{figure}
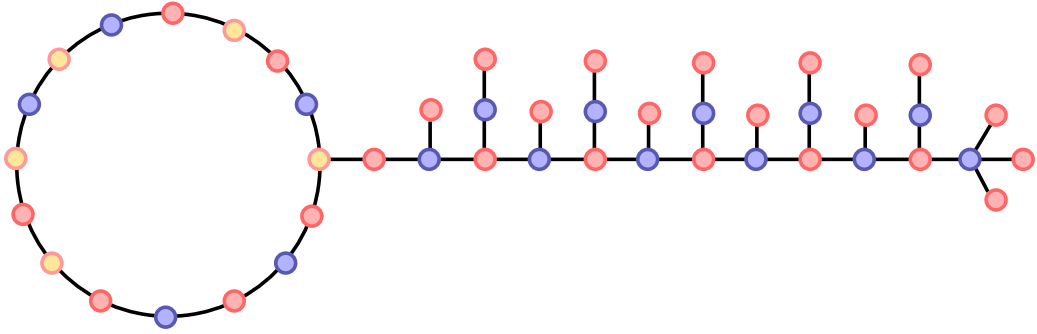
This algorithm uses at most $\log^*\!n$ rounds on the cycle and $v_1$, and constant time in the even-odd caterpillar. 
As the cyclic part has negligible size, the node-averaged complexity is constant.
\end{proof}

\section{Random ID assignments and randomized algorithms}

We move on to the second topic of this paper, where the randomized aspects are considered. 
The standard definition of the complexity in the $\LOCAL$ model not only considers the slowest node, but also the worst-case ID-assignment. 
In this section, we investigate the impact of replacing this measure by the running time (of the slowest node) on a random ID assignments. 
In other words, given a graph, we consider the average of the slowest-node running time over all possible ID assignments.

The main result is the equivalence between this measure, and the complexity of randomized algorithms for a classic class of problems. Here, the complexity of a randomized algorithm is the expectancy of the number of rounds before every node finishes, with an correct solution. 
Note that the two concepts have similar flavour, but are distinct. 
On the one hand, the random inputs of a randomized algorithm are independent, while in a random ID assignment, the identifiers are not independent. On the other hand, the IDs are distinct, while the random inputs can be equal.
On a high level, the equivalence is similar to Yao's principle \cite{Yao77}, that relates the performance of a randomized algorithm on a worst-case instance, and the complexity of a deterministic algorithm on a random instance. 
Also, note that in the literature, the usual complexity of randomized algorithms is not the one we consider, but the time needed to output a correct solution with high probability. That is, Monte Carlo algorithms  are considered instead of Las Vegas algorithms. We discuss briefly this point at the end of the section. 

\paragraph*{Random string length.} For the next theorem, randomized algorithms are given random strings of size $O(\log n)$, and not infinite such strings. 
This hypothesis is not excessive as most algorithm use a small amount of randomness. 
For example the celebrated MIS algorithm of \cite{Luby86} for bounded degree graphs, can be described as using random strings of size bounded by $O(\log n)$.  

\paragraph*{Completable $\LCLS$ languages.} The theorem deals with what we call \emph{completable $\LCLS$ languages}.\footnote{For this paragraph and the rest of this section we omit to mention possible inputs, as it does not have influence on the reasoning.} These are $\LCLS$ languages, with the following additional property. Consider a graph with some missing outputs, such that the verification algorithm accepts every neighbourhood that is fully labelled. Then if the language is completable there must exist a way to label the remaining nodes such that the resulting labelled graph is in the language is in the language. Note that the very classic problems of this area, such as $(\Delta+1)$-colouring or maximal independent set are completable. Also note that it is not the case of all the $\LCL$ problems, some of them such as sink-less orientation \cite{BrandtFHKLRSU16} are not completable.

\begin{theorem}\label{thm:random}
Given a completable problem in $\LCLS$, the expected slowest-node complexity of randomized algorithms, is asymptotically equal to the expected deterministic slowest-node complexity on identifier assignment taken uniformly at random.
\end{theorem} 

The proof is based on the well-known fact that randomized algorithms do not need IDs, because they can generate them. 
More precisely, it is folklore that taking $n$ integers in a cubic range uniformly at random, avoids collisions with high probability. 
The probability $1-1/n$ used in the later result is slightly too weak to be used in the proof of theorem \ref{thm:random}, so we make it $1-1/n^2$ with the following lemma. 

\begin{lemma}\label{lem:ncube} If $n$ numbers are taken independently  uniformly at random between 1 and $n^4$, these numbers are pairwise distinct with probability $1-1/n^2$.
\end{lemma}

\begin{proof}[Proof (Lemma \ref{lem:ncube})]
The probability of two fixed numbers being equal is $1/n^4$. Then by union bound, the probability that a pair of numbers have the same value is bounded by the number of such pairs $n(n-1)/2$ multiplied by the former probability. Then the probability of collision is bounded by $1/n^2$, thus with probability $1-1/n^2$ the numbers are pairwise distinct. 
\end{proof}

\begin{proof}[Proof (Theorem \ref{thm:random})]
As a preamble for the proof, remember that a randomized algorithm can be formalized as a deterministic algorithm having an auxiliary input, this input being a large enough random number. We consider an algorithm $A$ with an auxiliary input that can either be the ID or the random bits, and show that with high probability the behaviour is the same. As stated in lemma \ref{lem:ncube}, taking independently and uniformly at random $n$ numbers from $[n^4]$ provides a list of distinct numbers with probability $1-1/n^2$. Also when this sampling succeeds, that is when the numbers are distinct, the outcome is uniform among all distinct identifiers assignments, because the identifiers are taken independently uniformly at random.

Also note that in the context of completable $\LCLS$, given a graph containing some nodes with no output, it is always possible to compute a canonical completion of the current labeling. 
Indeed, one can for example choose an arbitrary order of the possible outputs for a node, order the outputs for the whole graph using a lexicographic order based on the identifiers, and then choose the smallest correct completion in this ordering to be the canonical completion.

Finally, let $t$ be the running time of the verification algorithm.

\paragraph*{From deterministic to randomized.} Let $D$ be a deterministic algorithm, and let $c$ be its expected slowest-node complexity on identifier assignments taken uniformly at random. 
Let $R$ be a randomized algorithm with the following behaviour on a node~$v$. 
It first picks a random number in $[n^4]$, and then runs $D$ with this pseudo-identifiers. 
If no collision is detected, then after some number of rounds the node knows what $D$ would output. 
Let us call this time~$r_D(v)$. The node then runs for $2t$ additional rounds, and if there is still no collision detected it outputs the same label as~$D$. This kind of output is called a \emph{regular output} in the rest of the proof. 
Otherwise, that is if a collision is detected, $v$ runs until it sees the whole graph and all the regular outputs. It then outputs the canonical output for that partially labelled graph.

For this algorithm to be well-defined we need to make sure that the partially labelled graph matches the definition of completable $\LCLS$ language. 
That is, the verifier must accept on every node that has a view without unlabelled vertices. Consider a node $v$ and its $2t$-neighbourhood $N$, that we suppose to be labelled. 
For every node in this $2t$-neighbourhood, we can consider the view that corresponds to its output, that is the view it had when $D$ stopped in the simulation. Let $S$ be the union of these views. 
If $S$ does not contain twice the same identifier, then we claim that the verifier accepts on $v$. This is because we can make sure to have no collision in the whole graph by changing all the identifiers outside of $S$. Then we would have a proper identified graph, on which $D$ produces correct outputs, and in particular it would produce the same output on the $2t$-neighbourhood of $v$, thus the verifier accepts on $v$. 
Suppose now that $S$ contains twice the same identifier. Let $a$ and $b$ be two nodes with the same ID in $S$. Then there exist two nodes $u$ and $w$ in $N$, such that $a$ is in the $r_D(u)$-view of $u$ and $b$ is in the $r_D(w)$-view of $w$. As $u$ and $w$ are at most $2t$ edges apart, then $\max(r_D(u)+r_D(w))+2t \geq \min(r_D(u)+r_D(w))$. This implies that when running for $2t$ additional rounds one of the two nodes would see the repeated ID. This is a contradiction, as such a node would then remain unlabelled for the first phase.

The algorithm is then well-defined and in addition it is correct, as it is based on a correct algorithm, with a completion that is consistent. 

Finally, the algorithm $R$ has probability at least $(1-1/n^2)$ to stop just $2t$ rounds after $D$ on every node, and probability at most $1/n^2$ to stop after at most $n$ rounds on some nodes. Then the expected runtime of the slowest node is upper bounded by $(1-1/n^2)(c+2t)+1/n^2.n$ which is asymptotically in $O(c)$.

\paragraph*{From randomized to deterministic.} Conversely, suppose that a randomized algorithm has expected complexity $c$. We claim that the same algorithm using the identifiers as random strings provides a deterministic algorithm with average complexity $c$. 
Let $c$ be the expected complexity when the random strings are all distinct, and $c'$ when they are non-distinct. 
The expected runtime is $(1-1/n^2)\cdot c+1/n^2\cdot c'$, which is asymptotically $c$, as $c'$ can be assumed to be at most $n$. Then the complexity is the same for the deterministic algorithm as for the randomized one.

Thus theorem \ref{thm:random} holds.
\end{proof}

A similar result can be obtained for the more classic context of Monte-Carlo algorithm. That is, when one considers the time before the nodes have stopped and output a proper solution with high probability, then the complexity of randomized algorithms and of deterministic algorithm on random identifiers are the same. 

A related topic is to minimize the amount of randomness used by randomized algorithms.
The amount of random bits necessary to perform a computation is usually not considered as a resource to be minimized in the $\LOCAL$ model, although it is in centralized computing, see \cite{PettieR08} for a precise example. Here, it is possible to do a small step in that direction, if we consider algorithms and languages that are local. In this case, it is not necessary to have all IDs of the graph that are different one from the other. In a local algorithm, the nodes see only a small neighbourhood of the graph, and thus only such neighbourhoods need to have distinct IDs. This is one of the ingredient of recent breakthroughs in the field, such as the speed-up theorem from \cite{ChangKP16}.\footnote{See theorem 6 in the paper.}

Let $s$ be the maximum number of nodes that a node can see when it runs the local algorithm at hand. Then the following holds:

\begin{proposition}
Taking uniformly at random numbers from $\left[n^2s^2\right]$ is sufficient to have locally distinct identifiers with high probability.
\end{proposition}

\begin{proof}
Consider a ball of size $s$. The probability that two nodes of this ball have the same identifier is upper bounded by $s^2/(n^2s^2)=1/n^2$. Then by union bound on all the centres of balls, one gets a probability of collision of $1/n$.
\end{proof}

\subsection{Node-averaged randomized complexity}\label{subsec:node_average_randomized}

After considering an average on the nodes, and on the identifiers assignment separately, we consider both averages together. That is we consider the behaviour of an ordinary node on an ordinary ID assignment. 
In the light of the previous subsection, this is equivalent to consider node-averaged complexity of randomized algorithms. This new measure can be unexpectedly low, as we illustrate on the example 3-colouring. 

Theorem \ref{thm:LCL_average} implies that the node-averaged complexity of 3-colouring of a cycle is $\Theta(\log^*\!n)$. It is also known that the randomized complexity is $\Theta(\log^*\!n)$, if one considers Monte-Carlo algorithms with probability of success greater than one half \cite{Naor91}. Then the expected running time is also in $\Theta(\log^*\!n)$. This contrasts with the following result.

\begin{proposition}
For 3-colouring on a ring, the expected complexity of an ordinary node is constant.
\end{proposition}

\begin{proof}
The algorithm we consider, consists in repeating a simple procedure. At each round every node that has not yet an output, take a colour at random among the colours that are still available. That is, it takes a colour that as not yet been output by a neighbour. Note that this is always possible, as the nodes have degree two, and choose among three colours. After the sampling, if there is no conflict, then the node outputs the colour. If there is a conflict, then the colour is forgotten, and the node continues to the next round. At a given round, if an uncoloured node outputs a colour, we say that it \emph{succeeds}, otherwise it \emph{fails}.

Given an arbitrary partial colouring obtained after some rounds, the probability that a fixed node succeeds is lower bounded by $\alpha=5/12$. 
This number is obtained by case analysis. 
It corresponds to the case where, the current node has both neighbours without outputs, but both nodes at distance two with outputs, and these outputs are different. Let $\beta=1-\alpha$. Also, let $V_k$ be the number of nodes that have not yet output after round $k$, with $V_0=n$. The following holds by linearity of the expectation.
 
\[
\EE(|V_k|\ | \ |V_{k-1}|)= \sum_{v\in V_{k-1}} \PP(v \text{ does $not$ stop at round $k$}) \leq \beta |V_{k-1}|
\]

We can apply the previous inequality repeatedly, and get: $\EE(V_k) \leq \beta^k n $. The number of nodes that stop at round $k$ is precisely $V_k-V_{k-1}$, then the sum of the running times is:
\[
\sum_k k(V_k-V_{k-1})\leq \sum_k kV_k.
\]
The expected sum of the running time is then upper bounded by $\sum_k k\beta^k n$. Then the node-averaged expected sum is $\sum_k k\beta^k $. As $\beta<1$, $\sum_k k\beta^k $ is a constant, thus the expected complexity of an ordinary node in a random ID assignment is constant.
\end{proof}

Note that having a constant complexity when looking at a more local measure, is not particular to this example. For example in \cite{Ghaffari16}, the author designs an algorithm for maximal independent set that terminates after $O(\log deg(v)+\log(1/\epsilon))$ rounds, with probability at least $1-\epsilon$, where $deg(v)$ is the degree of node $v$.  

\section{Conclusion and open questions}
This paper introduces the notions of node-averaged and ID-averaged complexities. We think these measures are meaningful when analysing distributed algorithms that do not have the knowledge of the size of the network, or in contexts where partial solutions are useful. Also, very local complexities, as the one of subsection \ref{subsec:node_average_randomized} and the one advocated in \cite{Ghaffari16}, are natural measures that one would like to understand better. Our results illustrate that these complexities can have interesting behaviours. An aspect that is not very satisfying is the assumption of linearly bounded growth in the local average lemma, it would be very interesting to know if this is necessary or not.

\section{Acknowledgements}
I would like to thank Juho Hirvonen, Tuomo Lempi\"ainen and Jukka Suomela for fruitful discussions, and Pierre Fraigniaud for both discussions, and help for the writing. I thank the reviewers for helpful comments, Patrice Ossona de Mendez for suggesting the name ``linearly bounded growth'', and Mohsen Ghaffari for pointing out that randomized  node-averaged complexity could be considered.  

This work is an extended and revised version of a preliminary conference
report \cite{Feuilloley17}. Part of the content is based on an earlier brief announcement \cite{Feuilloley15}.

The author received additional support from ANR project DESCARTES, and Inria project GANG.

\newpage

\DeclareUrlCommand{\Doi}{\urlstyle{same}}
\renewcommand{\doi}[1]{\href{http://dx.doi.org/#1}{\footnotesize\sf doi:\Doi{#1}}}
\bibliography{bibliography.bib}{}
\bibliographystyle{plainnat}

\end{document}